\documentclass[letterpaper,11pt]{article}
\usepackage[T1]{fontenc}
\usepackage[utf8]{inputenc}
\usepackage{lmodern}
\usepackage{microtype}
\usepackage{complexity}
\usepackage{amsmath}
\usepackage{amssymb}
\usepackage{amsthm}
\usepackage{fullpage}
\usepackage{graphicx}
\usepackage[export]{adjustbox} 
\usepackage{wrapfig}
\usepackage{xspace}
\usepackage[unicode=true]{hyperref}
\usepackage{braket}
\usepackage{algorithm}
\usepackage[noend]{algorithmic}

\newtheorem{theorem}{Theorem}
\newtheorem{lemma}{Lemma}
\newtheorem{corollary}{Corollary}
\theoremstyle{definition}

\newtheorem{claim}{Claim}
\newcommand{\opt}{\text{opt}}
\newcommand{\OPT}{\text{OPT}}
\newcommand{\LB}{\text{lb}}
\renewcommand{\E}{\mathbb{E}}

\usepackage[textsize=tiny,textwidth=2cm,color=green!50!gray]{todonotes} 

\newcommand{\cvrpmd}{\textsc{CVRP-MD}\xspace}
\newcommand{\cvrp}{\textsc{CVRP}\xspace}

\title{Approximating Multiple-Depot Capacitated Vehicle Routing\\ via LP Rounding}
\author{
    \begin{minipage}[t]{.3\textwidth}
        \centering
        Zachary Friggstad\thanks{Research supported by an NSERC Discovery Grant.}\\
        University of Alberta
    \end{minipage}
    \begin{minipage}[t]{.3\textwidth}
        \centering
        Tobias Mömke\thanks{Partially supported by DFG Grant 439522729 (Heisenberg-Grant).}\\
        University of Augsburg
    \end{minipage}
}
\date{}

\begin{document}
\maketitle

\begin{abstract}
In \emph{Capacitated Vehicle Routing with Multiple Depots} (\cvrpmd) we are given a set of client locations $C$ and a set of depots $R$ located in a metric space with costs $c(i,j)$ between $u,v \in C \cup R$. Additionally, we are given a capacity bound $k$.
The goal is to find a collection of tours of minimum total cost such that each tour starts and ends at some depot $r \in R$ and includes at most $k$ clients and such that each client lies on at least one tour.
Our main result is a $3.9365$-approximation based on rounding a new LP relaxation for \cvrpmd.
\end{abstract}


\section{Introduction}

In \emph{Capacitated Vehicle Routing with Multiple Depots} (\cvrpmd) we are given a set of client locations $C$ and a set of depots $R$ located in a metric space with costs $c(i,j)$ between $u,v \in C \cup R$. Additionally, we are given a capacity bound $k$.
The goal is to find a collection of tours of minimum total cost such that each tour starts and ends at some depot $r \in R$ and includes at most $k$ clients and such that each client lies on at least one tour.

The problem is well-studied in the Operations Research literature, as it has relevant applications such as delivering goods from multiple production facilities to consumers. For an overview of applied results we refer to \cite{MJI+15}. Notably, in the overview they write that most exact algorithms for solving the classical vehicle routing are difficult to be adapted for
solving \cvrpmd.

As \cvrpmd generalizes capacitated vehicle routing (\cvrp), it is APX-hard even for $k=3$~\cite{asano1997covering}. We therefore focus on approximation algorithms.
The first approximation algorithm for \cvrpmd was a $(2\alpha+1)$-approximation by Li and Simchi-Levi~\cite{li1990worst} where $\alpha$ is the approximation guarantee for the Traveling Salesman Problem (TSP) in the corresponding metric. Later, a $4$-approximation using a completely different approach was given by Harks, König and Matuschke~\cite{HKM13} in the course of studying a more general problem. We give a more detailed overview of these algorithms in Section~\ref{sec:review} as they are both relevant to our results.

Combining modern TSP results~\cite{karlin2021slightly,OG2} with \cite{li1990worst} gives a $(4 - 10^{-36})$-approximation for \cvrpmd.  To date, the best \cvrpmd approximation is by  Zhao an Xiao~\cite{ZX23}. The authors combine a recent approximation algorithm for \cvrp of Blauth, Traub and Vygen~\cite{blauth2021improving} with the results above and obtain an approximation ratio of $4 - 2 \cdot \delta \approx 3.9993$ where $\delta \approx 1/3000$ is a constant used in the work of Blauth, Traub, and Vygen.

Our main result is a new approximation algorithm with improved approximation ratio.
\begin{theorem}\label{thm:main}
There is a polynomial-time randomized algorithm for \cvrpmd that finds a solution whose expected cost is at most $3.9365$ times the optimum solution's cost.
\end{theorem}
Our improved approximation algorithm combines several key ideas from previous work in \cvrpmd plus branching decompositions of preflows~\cite{BJ95}. In fact, our new approach is compatible with the approach by Blauth, Traub, and Vygen and can be slightly improved by an amount in the order of $10^{-3}$ just like the \cvrpmd improvement in \cite{ZX23}.


\subsection{Notation and Preliminaries}
We view the given metric as a complete graph with vertices $V = C \cup R$ and edge distances given by $c(u,v), u,v \in C \cup R$. Throughout, we let $\OPT$ denote an optimal solution and $\opt$ the optimum solution cost. For a collection of edges $S$ (perhaps with multiple copies of an edge) we let $c(S)$ denote the total cost of all copies of edges in $S$. 
We view a tour $T$ as the multiset of edges, if $T$ only visits one client then it would contain two copies of the edge from the depot to the client.
For brevity, we let $r_T \in R$ denote the root of a tour $T$. For any multiset of edges $F$, we let $C_F$ denote all clients that are the endpoint of at least one edge in $F$. We will only use this notation when $F$ is a tour or a tree.

Most CVRP approximations in general metrics rely on the following \emph{radial lower bound}:
\[ \LB := \frac{2}{k} \sum_{v \in C} c(v,R) \]
where $c(v,R)$ denotes $\min_{r \in R} c(v,r)$. 
The following was proven by Haimovich and Rinnooy Kan~\cite{haimovich1985bounds} for the single-depot case, the proof extends immediately to the multiple-depot case as noted in \cite{li1990worst}.

\begin{lemma}[Haimovich, Rinnooy Kan~\cite{haimovich1985bounds}]\label{lem:lb_old}
$\LB \leq \opt$
\end{lemma}
\begin{proof}
Consider a tour $T$ in the optimum solution. For each $v \in C_T$, if we shortcut $T$ past all nodes except $r$ and $v$ then we see $2 \cdot c(v,R) \leq 2 \cdot c(v,r) \leq 2 \cdot c(T)$. 
Since $|C_T| \leq k$, we have $\frac{2}{k} \sum_{v \in C_V} c(v,R) \leq c(T)$. 
Summing over all tours $T$ in the optimum solution and recalling each client lies on a tour shows $\LB = \frac{2}{k} \sum_{v \in C} c(v,R) \leq \opt$.
\end{proof}

For each client $v \in C$, we say that $\frac{2}{k} \cdot c(v,R)$ is the radial lower bound contribution of $v$.


\subsection{Review of Previous Approximations}\label{sec:review}
Our approach incorporates elements of two previous approximations. We briefly summarize them below plus the way that the \cvrp algorithm by Blauth, Traub, and Vygen~\cite{blauth2021improving} was adapted to get the slightly improved approximation for \cvrpmd by Zhao and Xiao \cite{ZX23}.

~

\noindent
{\bf Tree Splitting} by Harks, König and Matuschke~\cite{HKM13}.\\
The algorithm starts with computing the minimum-cost $R$-rooted spanning forest, i.e., a collection of trees spanning $C \cup R$ where each tree contains a single $r \in R$ which we view as the root of that tree. 
Note that the cost of this forest is at most $\opt$ since the optimum solution, viewed as a collection of edges, ensures that each client is in the same component as some depot. If some tree $T$ in the forest includes more than $k$ clients we repeat the following until it does not.

First, pick the deepest node $v$ of $T$ whose subtree contains more than $k$ clients. Let $u_1, \ldots, u_{a}$ be the children of $v$, so each contains at most $k$ nodes in their subtree by our choice of $v$. 
Group these subtrees together arbitrarily so each group has at most $k$ nodes but two groups cannot be merged into a single group (i.e., any two groups together have more than $k$ clients). 
There must be at least one group with at least $k/2$ clients: otherwise, if there were two groups, they could be merged and if there is one group with less than $k/2$ clients then the total size of the subtree rooted at $v$ (including $v$ itself) is an integer less than $k/2 + 1$, so it must be at most $k$ which contradicts our choice of $v$.

For any group with at least $k/2$ clients, we turn it into a tour spanning all clients in the trees by doubling the trees including the parent edges to $v$. 
Graft this tour into the depots by picking a client $v'$ on the tour uniformly at random and adding the two copies of the edge $v'r$, where $r$ is the depot nearest to $v'$. 
Now shortcut the resulting Eulerian tour past repeated occurrences of nodes and also past {\em all} occurrences of $v$ (so $v$ is not on the tour). 
Also remove the subtrees of these groups from the tree $T$.
After repeating this pruning step until $T$ has at most $k$ clients, we get one last tour by doubling all edges on $T$ and letting the root of $T$ serve as the depot for this tour.

For a tour that was split off of the tree, each client $v'$ has a probability of at most $2/k$ of being picked to connect to its nearest depot $r$. 
Two copies of the edge $v'r$ were added in this case, so the expected total cost of all edges added in this way is at most $\frac{2}{k} \cdot \sum_{v' \in C} 2 \cdot c(v',R) = 2 \cdot \LB$.  Since the forest has cost at most $\opt$ and since we only doubled its edges to form tours, the total cost of the final solution is at most $2 \cdot \LB + 2 \cdot \opt \leq 4 \cdot \opt$ in expectation.  The algorithm is easy to derandomize by picking the nearest depot to each tour that is split off rather than of picking a client randomly and connecting it to its nearest depot.

~

\noindent
{\bf Tour Splitting} by Li and Simchi-Levi~\cite{li1990worst}\\
This algorithm is more along the line of the classic tour-splitting algorithm for single-depot CVRP. 
It starts by considering the metric obtained by contracting all of $R$ to a single node $v_R$ and recomputing the shortest path distances. 
It uses an $\alpha$-approximation for TSP in this contracted metric. 
Replacing each edge of the TSP tour by its shortest path in the contracted graph yields a collection of $v_R$-rooted tours. 
These tours lift to walks in the original graph whose endpoints lie in $R$ (some walks may start and end at different depots). The cost of these walks is at most $\alpha \cdot \opt$
and each client lies on at least one walk. 
Next, shortcut each of these and remove the endpoint (a depot) to get a collection of paths $P$, each starting at some $r \in R$ and each client lying on exactly one path.

Finally, standard tour partition is performed: for each such path $P$ we pick every $k$-th client along $P$ starting with an initial random offset from the depot on $P$. 
For each chosen client $v$ a tour is formed by starting at the depot closest to $v$ and visiting the next $k-1$ clients after $v$ on $P$ (or to the end of $P$ if there are fewer clients). The remaining clients before the first one that is picked are connected to the start node $r'$ of $P$ (a depot) by doubling the subpath from $r'$ to these remaining clients.

The total cost of the doubled subpaths is at most $2 \alpha \cdot \opt$ and the total cost of the doubled edges connecting sampled clients to their nearest depot is at most $\frac{2}{k} \cdot \sum_{v \in C} c(v,R) = \LB$ in expectation. 
The total cost is then at most $2\alpha \cdot \opt + \LB$ which, given the current state of the art for TSP, is at most $(4 - 10^{-36}) \cdot \opt$ \cite{karlin2021slightly,OG2}.

~

\noindent{\bf Combining Approaches}\\
Zhao and Xiao~\cite{ZX23} note that the improved algorithm for \textsc{CVRP} by Blauth, Traub and Vygen~\cite{blauth2021improving} yields an improved approximation for \cvrpmd as follows. 
If $\LB \leq (1-\delta) \cdot \opt$ for some constant $\delta > 0$ then the tree splitting approach is already better\footnote{As is tour splitting, but the improvement is greater with tree splitting given the $1.5-10^{-36}$ TSP approximation}. 
Otherwise, in the contracted metric from the tour splitting approach it is possible to use the algorithm by \cite{blauth2021improving} to get a TSP solution whose cost is at most $(1+f(\delta)) \cdot \opt$ where $f(\delta)$ is some function that vanishes as $\delta \rightarrow 0$. 
Following the rest of the tour splitting approach finds a solution with cost at most $2 \cdot (1 + f(\delta)) \cdot \opt + \LB \leq (3 + 2 \cdot f(\delta)) \cdot \opt$. 
By picking $\delta = 1/3000$ as in \cite{blauth2021improving} they get a $(4-1/1500)$-approximation.


\subsection{Our Approach}

Our approach combines elements of tree splitting and tour splitting along with a bidirected cut relaxation for the problem. First, notice that the tour splitting approach would yield a much better approximation if it was possible to get a TSP tour whose cost is considerably better than $\alpha \cdot \opt$. That is essentially what is done in \cite{ZX23} to get their $(4-1/1500)$-approximation but in their approach the tradeoff between the two approaches based on how close $\LB$ is to $\opt$ requires $\LB$ to be very close to $\opt$ to see any improvement in the final approximation guarantee.

We do considerably better in this tradeoff, but at the expense of not spanning all nodes.
For now, let $\delta$ be such that $\LB = (1-\delta) \cdot \opt$. Our actual definition of $\delta$ is slightly different but the above simple definition serves to explain the high level approach we take.

That is, using an LP relaxation and a preflow-splitting result by Bang-Jensen et al.~\cite{BJ95}, we sample a collection of paths rooted at the depots whose expected cost is $(1+\delta) \cdot \gamma \cdot \opt$ for some constant $\gamma$ close to $1/2$. These paths may not span all nodes, but we ensure any client $v \in C$ has a probability at most $e^{-\gamma}$ of not being spanned. Using an adaptation of the Bridge Lemma for vehicle routing problems in \cite{BFMS25}, we can graft in the missing clients into the sampled paths with a forest whose cost is at most $e^{-\gamma} \cdot \opt$ in expectation.\footnote{The approach in \cite{BFMS25} is to use the preflow-splitting result of \cite{BJ95} and to apply an argument related to the bridge lemma of \cite{BGRS13} in order to collect vertices not covered by branchings. In our new result, we use a similar starting point, but we obtain the preflows from a more involved LP and use the result from \cite{BFMS25} as a black box for collecting the missing vertices. A similar approach of sampling a tree and grafting in missing nodes is also taken in \cite{AMN24}.
}

We are therefore left with $R$-rooted paths spanning some clients and a forest of trees that grafts the missing nodes into these paths. Our high-level goal is to turn these paths and trees into feasible tours. 
We introduce extensions of both tree splitting and tour splitting to obtain these tours while charging their cost to:
\begin{itemize}
    \item At most $2$ times the cost of all edges on the paths.
    \item At most $1$ times the radial lower bound contribution of clients covered by the paths.
    \item At most $2.5$ times the cost of the forest.
    \item At most $1.694323$ times the radial lower bound contribution for clients covered by the forest.
\end{itemize}
Some intricate cancellation happens: recall that in expectation, the cost of the paths is at most $\gamma \cdot (1+\delta) \cdot \opt$ where $\gamma < 1/2$ and the cost of the forest is at most $e^{-\gamma} \cdot \opt$. All clients' radial lower bounds are charged at least once, and an additional at most $0.694323$ times if they are not covered by the paths, which happens only with probability at most $e^{-\gamma}$. Using $\LB \leq (1-\delta) \cdot \opt$, a calculation shows that the positive dependence on $\delta$ charged to the edge costs is overshadowed by the negative dependence on $\delta$ in the radial lower bound charges. For the optimal choice of $\gamma$, this leaves us with an approximation guarantee of $3.9365 - 0.49826 \cdot \delta$.

~

\noindent
{\bf New Tree Splitting}\\
For our tree pruning procedure, we first consider pruning a group of subtrees if the cost of the tour can be charged to $2.5$ times the edge cost in these subtrees plus $1.694323$ times the radial lower bound of the clients. If this is not possible, we consider a new step that tries to cover exactly $k$ clients by completely covering all clients in one group of subtrees and an appropriate number of carefully chosen clients in another subtree: if this can be done in a way that admits the same charging scheme to edges in the completely-removed subtrees and radial lower bound contribution of covered clients, we do it. We show that when this is not possible, the tree has fewer than $k$ clients plus additional properties on its costs that enable the next step.

~

\noindent
{\bf New Tour Splitting}\\
Once all trees are ``small'' after tree pruning, we consider each path $P$ and the short trees rooted in $P$. Using a careful DFS ordering of the trees, we obtain a single path $P'$ that is split in a manner similar to \cite{li1990worst}, except additional care is taken. That is, if a tour split occurs within one of the short trees rooted in $P$ then some of the edges of the tree would be charged more than twice in the final tours. We take a more careful approach to how we deal with a split in the middle of a tree to limit the extent to which its edges are charged more than twice, this is enabled through properties of the short trees that are produced by the tree pruning step and may result in some clients in the short trees being selected more than once to connect to their nearest depot. Ultimately, in this step we still ensure that each edge of a small tree is charged at most $2.5$ times and each client covered by these trees has its radial lower bound charged to an extent of $1.69432$, just like with the tree pruning phase. Furthermore, the edges on the paths and radial lower bounds of clients on the paths will be charged to an extent of $2$ and $1$, respectively.

We finish this overview by remarking that our approach is compatible with the Blauth, Traub, and Vygen approach to \textsc{CVRP} \cite{blauth2021improving}. 
That is, if $\LB$ is smaller than $\opt$ by some constant factor then our new algorithm is even better than our worst-case analysis and if $\LB$ is very close to $\opt$ then an extension of \cite{blauth2021improving} to \cvrpmd described in \cite{ZX23} will also give an approximation guarantee better than our algorithm. So our final approximation guarantee can further be improved by a small constant in the order of $10^{-3}$.


\section{\cvrpmd Algorithm}\label{sec:alg}

We begin by giving an LP relaxation and describing how to sample paths covering some clients. After this is presented, in Section \ref{sec:summary} we provide pseudocode for the entire algorithm which calls on two subroutines that are fully described and analyzed in later sections.

We can assume $k \geq 3$ as the case $k = 1$ is trivial and the case $k = 2$ can be solved by a minimum-cost matching algorithm (the cost of a matching between clients being the cheapest tour that serves them both) allowing for a client to be matched with itself.

We also assume $c(v,R) > 0$ for each client $v \in C$ as we could trivially cover clients with $c(v,R) = 0$ using 0-cost singleton tours. For a client $v \in C$ we let $\ell_v = \frac{2}{k} \cdot c(v,R)$ denote its radial lower bound contribution and for $C' \subseteq C$ we let $\ell(C') = \sum_{v \in C} \ell_v$ be the contribution of clients in $C'$ to $\LB$. In particular, $\ell(C) = \LB$.

Consider the following LP relaxation in the bidirected graph obtained by considering both directions of any edge to be different directed edges. Here, variable $z^r_{v,v}$ indicates that there is a tour from depot $r \in R$ that has $v$ as its furthest node (breaking ties arbitrarily) and $z^r_{v,u}$ for $u \neq v$ indicates $u$ is covered by the corresponding tour from $r$ that has $v$ as its furthest node.
We view all such tours as being two $r-v$ paths and we model them in the LP as $2$ units of flow from $r$ to $v$ in the directed graph. For a directed edge $e$, let $x^r_{v,e}$ indicate that edge $e$ is traversed in the forward direction among the $2$ paths from $r$ to $v$.

\begin{alignat*}{2}
{\bf minimize}\quad \sum_{r \in R} \sum_{v \in C} \sum_e c(e) x^r_{v,e}&&& \label{lp:cvrp} \tag{{\bf LP-CVRP-MD}}\\
{\bf s.t.}\qquad \qquad
 x^r_v(\delta^{out}(u)) &=  \left\{\begin{array}{rl} 2 \cdot z^r_{v,v} & \text{ if } u = r \\ 0 & \text{ if } u = v \\ z^r_{v,u} & \text{ otherwise}\end{array}\right. && \forall~ r \in R, u,v \in V \\
 x^r_v(\delta^{in}(u)) &=  \left\{\begin{array}{rl} 0 & \text{ if } u = r \\ 2 \cdot z^r_{v,v} & \text{ if } u = v \\ z^r_{v,u} & \text{ otherwise}\end{array}\right. && \forall~ r \in R, u,v \in V \\
 x^r_v(\delta^{in}(S)) &\geq  z^r_{v,u} && \forall~ r \in R, u,v \in V, \{u\} \subseteq S \subseteq V-\{r\} \\
 \sum_{r \in R} \sum_{v \in C} z^r_{v,u} &=  1 && \forall~ u \in V \\
 z^r_{v,u} &\leq  z^r_{v,v} && \forall~ r \in R, u,v \in C \\
 z^r_{v,u} &=  0 && \forall~ r \in R, u,v \in C \text{ s.t. }c(u,r) > c(v,r) \\
 \sum_{u \in C} z^r_{v,u} &\leq  k \cdot z^r_{v,v}&& \forall~r \in R, v \in V\\
 x,z &\geq  0&&
\end{alignat*}

Here and throughout the rest of the description of the algorithm, we define $\delta$ to be such that $\sum_{r,v} 2 c(v,r) \cdot z^r_{v,v} = (1-\delta) \opt_{LP}$ where $\opt_{LP}$ refers to the optimum solution value of the LP relaxation. For each $r \in R, v \in C$, since the vector $x^r_v$ is sending $2z^r_{v,v}$ units of $r-v$ flow then its cost is at least $2 c(v,r)$. So $\delta \in [0,1]$.

\begin{lemma}\label{lem:lb}
$\opt_{LP} \leq \opt$ and $\LB \leq (1-\delta) \cdot \opt_{LP}$.
\end{lemma}
\begin{proof}
The natural $\{0,1\}$-solution corresponding to the optimum \cvrpmd solution is feasible. That is, each tour can be naturally viewed as 2 units of flow to the farthest node from the depot on that tour. It is routine to verify that all constraints hold under this assignment and that the cost of this $\{0,1\}$ solution is $\opt$.

For the second part of the lemma, consider any $r \in R, v \in C$ and any $u$ supported by the flow $x^r_v$, so $c(u,r) \leq c(v,r)$. Then the LP constraints imply
\[ 
\sum_{u \in C} z^r_{v,u} \cdot 2c(r,u)\leq k \cdot 2c(v,r) \cdot z^r_{v,v} = (1-\delta) \opt_{LP}.
\]
Summing over all $r \in R, v \in C$ and using $\sum_{r,v} z^r_{v,u} = 1$ for each $u \in C$ completes the proof.
\end{proof}

Finally, we recall the following decomposition by Bang-Jensen, Frank, and Jackson~\cite{BJ95} which we paraphrase in a way that is convenient for us. An (out) $r$-branching in a directed graph is a tree rooted at $r$ that is oriented away from $r$ but does not necessarily span all nodes of the graph. Recall also that an $r$-preflow in a directed graph is an assignment $f$ of values to edges such that $f(\delta^{in}(v)) \geq f(\delta^{out}(v))$ for all nodes $v \neq r$.
\begin{theorem}\label{thm:bj}
Let $G = (V,E)$ be a directed graph. Let $f$ be an $r$-preflow with rational-valued entries for some $r \in V$. Suppose each $v \in V-\{r\}$ has $r\rightarrow v$ connectivity $\lambda_v$ under $f$. For any rational value $K \geq 0$, there are $r$-branchings $B_i$ with associated weights $\mu_i \geq 0$ with $\sum_i \mu_i = K$ such that each $e \in E$ lies on at most an $f_e$-weight of branchings and each $v \in V-\{r\}$ lies on at least a $\min\{K,\lambda_v\}$-weight of branchings.
\end{theorem}
Post and Swamy~\cite{PS15} showed how to find the branchings $\{B_i\}$ and corresponding weights $\{\mu_i\}$ in polynomial time. In particular, the decomposition involves only a polynomial number of branchings.


\subsection{Step: Preflow Rounding}

Consider some constant $0 < \gamma \leq 1/2$. We will eventually fix $\gamma := 0.46821$ to optimize our analysis. 
We use the following procedure to sample a collection of $R$-rooted paths from an optimal extreme point (hence rational) solution $(x,z)$ to \eqref{lp:cvrp}.
\begin{itemize}
\item For each $r \in R, v \in C$ consider the $r$-$v$ (pre)flow given by values $\gamma \cdot x^r_{v,e}$, which is sending $2\gamma \cdot z^r_{v,v}$ units of flow from $r$ to $v$. 
Using $K_{r,v} := 2\gamma \cdot z^r_{v,v}$ in Theorem \ref{thm:bj}, we obtain a collection of branchings $\{B_i\}_{i \ge 1}$ with corresponding nonnegative weights $\mu_i$ summing to $K_{r,v}$ such that each directed edges lies on at most a $\gamma \cdot x^r_{v,e}$-fraction of branchings. 
Furthermore, $v$ lies on each of these branchings and every other $u \in C-\{v\}$ lies on at least a $\gamma \cdot z^r_{v,u}$-fraction of these branchings (noting $\gamma \cdot z^r_{v,u} \leq K_{r,v}/2$ by the LP constraints).
Let $\mathcal B_{r,v}$ be the subset of these branchings $\{B_i\}_{i \ge 1}$ obtained by independently sampling each one to lie in $\mathcal B_{r,v}$ with probability $\mu_i$ which is at most $1$ since $\mu_i \leq K_{r,v}$ and $\gamma \leq 1/2$.
\item Next, for each $r \in R$ and $v \in C$ we turn each branching $B \in \mathcal B_{r,v}$ into an $r$-$v$ path by adding the reverse direction of all edges not on the $r$-$v$ path in $B$ and shortcutting the resulting Eulerian walk. Call the resulting path $P(B)$.
\item Finally, we simplify the resulting collection of paths as follows. While there is some $v \in C$ that lies on at least two paths $P(B), P(B')$ remove the occurrence of $v$ from one of these paths, choosing arbitrarily, by either shortcutting the path if $v$ was not the last node on the path or truncating the path to end at the node just before $v$ if $v$ was the last node on the path. Discard any paths that no longer contain any clients.
\end{itemize}
For $r \in R$ let $\mathcal P_r$ be the set of the resulting paths that start at $r$ and let $\mathcal P = \cup_{r \in R} \mathcal P_r$ be the set of all such paths.

The resulting collection of paths have the following properties: a) each path $P \in \mathcal P_r$ begins at $r$ and, otherwise, only contains clients, b) each path in $\mathcal P$ visits at least one client, c) no client lies on more than one path in $\mathcal P$ (though multiple paths may originate from the same depot), d) a client lies on a path in $\mathcal P$ if and only if it lies on at least one sampled branching.

\begin{lemma}\label{lem:path}
$\E[\sum_{P \in \mathcal P} c(P)] \leq \gamma \cdot (1+\delta) \cdot \opt_{LP}$
\end{lemma}
\begin{proof}
For a branching $B$ we have $c(P(B)) \leq 2 \cdot c(B) - c(v,r)$ because the path was obtained by doubling all edges {\em not} on the $r$-$v$ path in $B$ and then shortcutting the resulting Eulerian walk. So the expected cost of $\mathcal P$ is at most
\begin{eqnarray*}
    & & \sum_{r \in R} \sum_{v \in C} \sum_{B_i \in \mathcal B_{r,v}} \mu_i \cdot (2 \cdot c(B) - c(v,r)) \\
    & \leq & 2\cdot \gamma \cdot \opt_{LP} - \sum_{r \in R} \sum_{v \in C} c(v,r) \sum_{B_i \in \mathcal B_{r,v}} \mu_i \\
    & = & 2\cdot\gamma \cdot \opt_{LP} - \sum_{r \in R} \sum_{v \in C} K_{r,v} \cdot c(v,r)
\end{eqnarray*}
where the inequality is because for each $r \in R$ and $v \in C$ we have each edge $e$ lying on at most a $\gamma \cdot x^r_{v,e}$-fraction of branchings in $\mathcal B_{r,v}$.

Recalling $K_{r,v} = 2\gamma \cdot z^r_{v,v}$, the last term $\sum_r \sum_v K_{r,v} \cdot c(v,r)$ equals $\gamma \cdot \sum_{r,v} 2 c(r,v) \cdot z^r_{v,v} = (1-\delta) \cdot \opt_{LP}$ so the expected cost of $\mathcal P$ is seen to be at most
\[ 2\cdot \gamma \cdot \opt_{LP} - \gamma \cdot (1-\delta) \cdot \opt_{LP} = \gamma \cdot (1+\delta) \cdot \opt_{LP}.\]
\end{proof}


\subsection{Step: Covering the Uncovered Nodes with Trees Rooted in $\mathcal P$}

Recall that $\mathcal P$ is the set of all $R$-rooted paths sampled in the previous step. Let $U$ be all clients that do not lie on a path, this is the same as the set of clients that did not lie on any branching that was sampled.

The next step is to compute a cheapest spanning forest $F$ rooted at nodes in $R \cup (C-U)$. That is, each component of $F$ includes precisely one node in $R \cup (C-U)$. Intuitively, the trees of $F$ are being used to graft in uncovered clients into the paths.
Before converting $\mathcal P$ and $F$ into a feasible \cvrpmd solution, we note a few properties of $U$ and $F$.

\begin{lemma}\label{lem:penalty}
$\E[\ell(U)] \leq e^{-\gamma} \cdot \LB$
\end{lemma}
\begin{proof}
For each $r \in R, v \in V$ if we view each branching $B_i$ obtained from Theorem \ref{thm:bj} as a set covering its clients, then each client $u \in C$ lies on a total weight of $\gamma \cdot z^r_{v,u}$ of such sets. Summing over all $r,v$ we see that each $u \in C$ lies on a total weight of $\sum_{r,v} \gamma z^r_{v,u} = \gamma$ such sets. Since the branchings were sampled independently, the usual analysis for randomized-rounding set cover algorithms using the arithmetic-geometric mean inequality shows that a node $u \in C$ is not covered by any sampled branching with probability at most $e^{-\gamma}$.
\end{proof}

To bound the expected cost of $F$, we use the following adaptation of the Bridge Lemma of \cite{steiner} to vehicle routing problems.
\begin{theorem}[Böhm et al.~\cite{BFMS25}]\label{thm:bridge}
Let $G = (V,E)$ be an undirected graph with edge costs $c$. For any $S \subseteq V$, let $c_S$ be the minimum-cost of a spanning forest whose trees are rooted in $S$.
Let $T \subseteq V$ be nonempty and let $\mu : 2^V \rightarrow \mathbb{Q}_{\geq 0}$ be a probability distribution over subsets of $V \setminus T$. Also suppose there is some $\gamma \geq 0$ such that for that any $v \in V\setminus T$ we have ${\bf Pr}_{S \sim \mu}[v \notin S] \leq \gamma$. Then $\E[c_{T \cup S}] \leq \gamma \cdot c_T$.
\end{theorem}

Using this, we bound the expected cost of $F$ as follows.
\begin{lemma}\label{lem:bridge}
$\E[c(F)] \leq e^{-\gamma} \cdot \opt_{LP}$
\end{lemma}
\begin{proof}
In Theorem \ref{thm:bridge}, the graph is the undirected metric over $R \cup V$, $T := R$, and the probability distribution $\mu$ is given by the subset of clients that are covered by a path in $\mathcal P$. From the analysis of Lemma \ref{lem:penalty}, we know for each client $v$ that ${\bf Pr}_{S \sim \mu}[v \notin S] \leq e^{-\gamma}$. So by Theorem \ref{thm:bridge}, ${\bf E}[c(F)] \leq e^{-\gamma} \cdot c_R$.

We conclude by showing $c_R \leq \opt_{LP}$. To do that, consider the following natural cut-based  LP relaxation for the minimum-cost $R$-rooted directed spanning forest in the bidirected graph obtained from our given metric over $R \cup V$.
\[ \min\left\{\sum_e c(e) \cdot y_e : y(\delta^{in}(S)) \geq 1 ~\forall~ \emptyset \subsetneq S \subseteq V, y \geq 0\right\}. \]
Here, we are letting $y_e$ be a variable for each directed edge and $c(e)$, for a directed edge is the same as the cost of the underlying undirected edge.

It is well known that that extreme points of this relaxation are integer, e.g., by contracting $R$ to a single node an using integrality of the corresponding directed arborescence LP relaxation \cite{Edm67}.
It is also easy to verify that the vector $\sum_{r \in R, v \in C} x^r_v$ is feasible for this LP relaxation. So the cheapest $R$-rooted spanning forest has cost at most $\opt_{LP}$. That is, $c_R \leq \opt_{LP}$.
\end{proof}

Throughout the rest of the algorithm, we remove clients from the set $U$ as they are covered. In the description of the algorithm, we will say a client $u \in C$ is {\em uncovered} if $u \in U$ when that step is executed.


\subsection{Algorithm Summary}\label{sec:summary}

The entire algorithm for vehicle capacity $k \geq 3$ is summarized in Algorithm \ref{alg:full}. Here, we use $\gamma = 0.46821$. The two subroutines mentioned in Algorithm \ref{alg:full} will be covered in the sections below.

\begin{algorithm}
\caption{Multiple-Depot CVRP Algorithm}
\begin{algorithmic}\label{alg:full}
\STATE Let $(x,z)$ be an extreme point optimal solution to \eqref{lp:cvrp}.
\STATE $\mathcal P \gets \emptyset$
\FOR{each $r \in R, v \in C$}
\STATE Apply Theorem \ref{thm:bj} to preflow $\gamma \cdot x^r_v$ using $K := 2\gamma \cdot z^r_{v,v}$ to get branchings $\mathcal B_{r,v} = \{B_i\}$ with corresponding weights $\{\mu_i\}$. \COMMENT{Recall $v$ will be on each such $B_i$.}
\FOR{each $B_i \in \mathcal B_{r,v}$}
\STATE Let $P_i$ be an $r-v$ path obtained by doubling edges of $B_i$ not on the $r-v$ path in $B_i$ and shortcutting the resulting Eulerian $r-v$ walk past repeated clients.
\STATE With probability $\mu_i$ add $P_i$ to $\mathcal P$.
\ENDFOR
\ENDFOR
\STATE $U \gets$ clients in $C$ not lying on any path in $\mathcal P$.
\STATE Shortcut paths in $\mathcal P$ so each $v \in C-U$ lies on exactly one path.
\STATE $F \gets$ min-cost spanning forest of $R \cup C$ rooted in $R \cup (C-U)$.
\FOR{each tree $T$ of $F$}
\STATE Apply Algorithm \ref{alg:pruning} in Section \ref{sec:pruning} to cover some clients of $T$ with tours and prune it to a shorter tree $T'$.
\ENDFOR
\FOR{each path $P \in \mathcal P$}
\STATE Apply the tour partitioning procedure in Section \ref{sec:final_part} to $P$ and the trees $T'$ rooted in $P$ to cover the remaining clients.
\ENDFOR
\end{algorithmic}
\end{algorithm}

\subsection{Step: Tree Pruning}\label{sec:pruning}

Recall $F$ is a forest where each tree has a single node that is covered by the paths in $\mathcal P$. That is, each tree can be viewed as rooted in $R \cup (C-U)$. All other nodes in the tree are initially uncovered.
The goal of this section is to splice off some tours to cover all but a bounded number of uncovered clients in each tree $T$, i.e., we want each tree rooted in a node rooted in $\mathcal P$ to be {\em small}.

Each time we cover some clients in a tree $T$, we will cut out some subtrees of $T$ and declare some uncovered nodes of $T$ to now be covered. The cost of the tour will be charged to both the cost of the edges removed from $T$ and also to the radial lower bound for the newly-covered clients.

Throughout the algorithm, for a node $v$ in tree $T$ we let $T_v$ be the subtree rooted under $v$ and $U(T_v)$ be the uncovered nodes in $T_v$. For a node $u$ of $T$, let $u_1, u_2, \ldots$ denote its children.
If each $T_{u_i}$ has at most $k$ uncovered nodes, we can consider the following {\bf grouping scheme}. Initially let each $T_{u_i}$ lie in a group of its own. Then while it is possible to merge two groups so the resulting group still has at most $k$ uncovered clients, do so.
For such a group $G$, let $c(G)$ denote the total cost of all subtrees in the group plus the edges connecting these subtrees to $u$ and $U(G)$ be the uncovered clients lying on a subtree in $G$.

For a group $G$ we let $\mathcal T(G)$ be a tour obtained by doubling the edges on all subtrees in $G$ (including the edge connecting subtrees $T_{u_i}$ to $u$) and adding two copies of the cheapest edge connecting some depot $r' \in R$ to some client in $U(G)$, and then shortcutting the resulting Eulerian tour so it only visits $\{r'\} \cup U(G)$.

For two different groups $G, G'$ we also consider the following tour $\mathcal T(G,G')$ that covers exactly $k$ uncovered clients.
\begin{itemize}
\item Let $A$ be the $k-|U(G)|$ clients in $U(G')$ with largest $\ell_v$-values. This is possible since $|U(G)| + |U(G')| > k$ as $G$ and $G'$ cannot be merged.
\item Let $\mathcal T(G,G')$ be the tour spanning $U(G) \cup A$ and a depot $r'$ obtained by first doubling all edges in the subtrees of both $G$ and $G'$ along with the parent edges of these subtrees conneecting them to $u$, then adding two copies of the cheapest edge connecting some $r' \in R$ to some client in $U(G) \cup A$, and finally shortcutting the resulting Eulerian tour so it only visits $\{r'\} \cup U(G) \cup A$.
\end{itemize}

Finally, we introduce two constants: (a) $\beta = 0.5902302342$ governs the behaviour of the algorithm and (b) $\Delta = 1.6353454381$ is used in the analysis. These parameters were chosen in a way to optimize the approximation guarantee obtained from our analysis while satisfying the following bounds that will be used at various points in our arguments.
\begin{itemize}
\item $2 - 0.5/\Delta \leq 1/\beta$
\item $1.5 + (\beta-0.5)/\Delta \leq 1/\beta$
\item $1/(1-\beta) < 3$
\item $0 \leq \frac{\Delta - (1-\beta) \cdot 1.5}{3-2\beta} \leq 1$ and $0 \leq \frac{2 + 3\beta/\Delta}{2+1/(1-\beta)} \leq 1$.
\end{itemize}

Algorithm \ref{alg:pruning} summarizes the tree pruning procedure. Notice if the second case within the while loop is executed, the subtrees in $G_j$ are not pruned from the tree. Figure \ref{fig:prune} illustrates this. Also, this case is only executed if the cost of the resulting tour can be charged against only $c(G_i)$ and $\ell(U(G_i) \cup A)$, i.e. the edges in the subtrees from $G_j$ will not be charged.

\begin{algorithm}
\caption{Pruning a tree $T$ from the forest $F$ rooted at a node $v_T \in R \cup (C - U)$.}
\begin{algorithmic}\label{alg:pruning}
\WHILE{the following removes some clients from $U_T$}
\STATE Let $u$ be any deepest node in $T$ such that $|U(T_u)|> k$, using $u = v_T$ if no such node exists
\STATE Group the subtrees $T_{u_1}, \ldots$ into groups $G_1, G_2, \ldots$ using the procedure discussed above
\IF{$c(\mathcal T(G_i)) \leq 2.5 \cdot c(G_i) + \frac{1}{\beta} \cdot \ell(U(G_i))$ for some group $G_i$}
\STATE Include tour $\mathcal T(G_i)$ in the final answer
\STATE Remove all subtrees in $G_i$ from $T$
\STATE $U \gets U-U(G_i)$
\ELSIF {$c(\mathcal T(G_i, G_j)) \leq 2.5 \cdot c(G_i) + \frac{1}{\beta} \cdot \ell(U(G_i) \cup A)$ for two distinct groups $G_i, G_j$}
\STATE \COMMENT{Here, $A$ is the set of $k-|U(G_i)|$ clients in $U(G_j)$ with largest $\ell_v$-values}
\STATE Include tour $\mathcal T(G_i,G_j)$ in the final answer
\STATE Remove only subtree $G_i$ from $T$ \COMMENT{Clients in $A$ remain in $T$, but are still regarded as covered}
\STATE $U \gets U - (U(G_i) \cup A)$
\ENDIF
\ENDWHILE
\end{algorithmic}
\end{algorithm}

\begin{figure}[ht]
\begin{center}
\includegraphics[width=12cm]{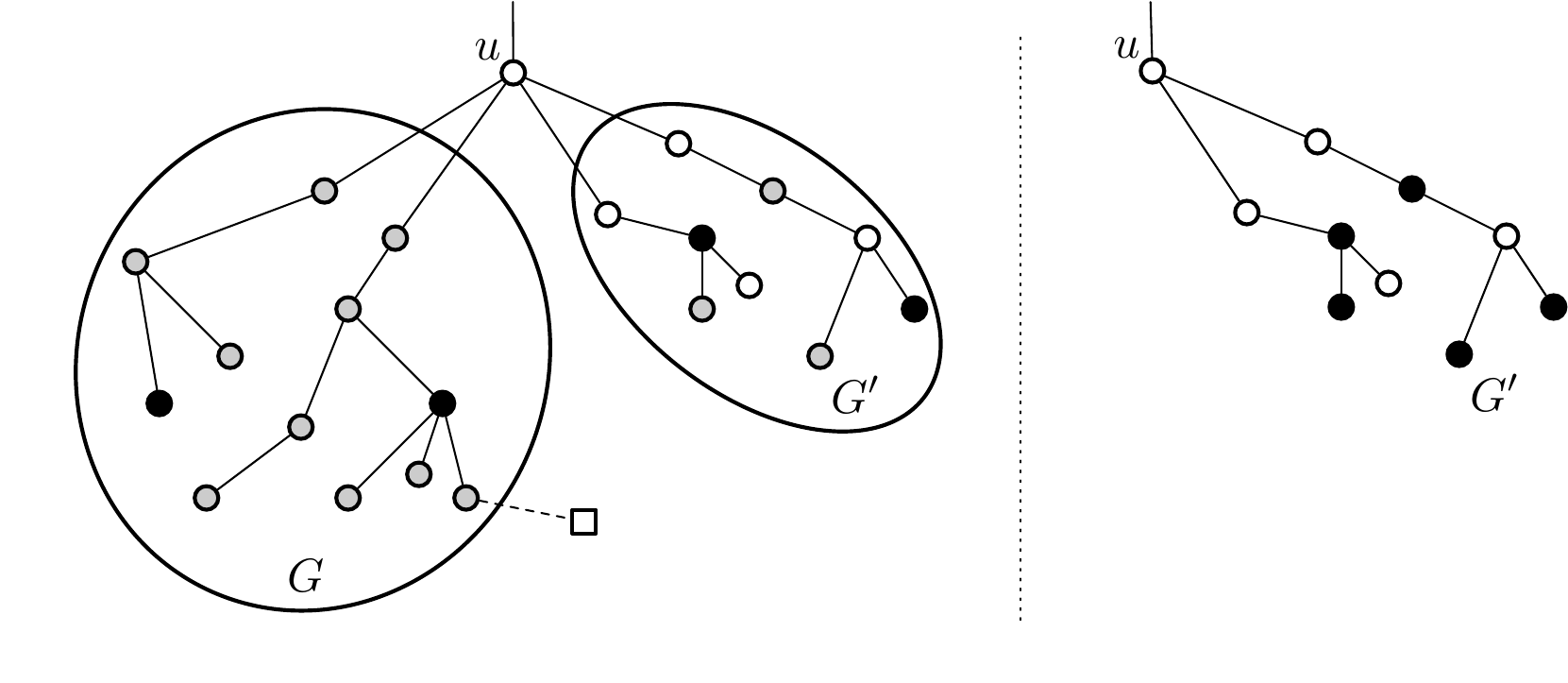}
\end{center}
\caption{{\bf Left}: The black nodes were already covered by a previous pruning step and the grey nodes are the ones to be covered in the tour $\mathcal T(G_i,G_j)$ (so $k = 10$). All nodes of $U(G)$ will be covered and the grey nodes in $G'$ are precisely $A$. The square node is the depot closest to a client in $U(G_i) \cup A$. Tour $\mathcal T(G_i,G_j)$ is obtained
by doubling all edges shown in the picture and shortcutting to only include $U(G_i) \cup A$ and the depot. {\bf Right}: All subtrees in $G_i$ are pruned, but all subtrees in $G_j$ remain. The nodes of $A$ are now covered.}
\label{fig:prune}
\end{figure}

In each iteration of the loop, the node $u$, while uncovered at that time, will not be considered covered by any tour that is found that iteration as such tours only cover nodes in subtrees rooted under children of $u$ even though the parent edge to $u$ of such a subtree is used in the construction of the tour.

The following is immediate from the algorithm.
\begin{corollary}\label{cor:pruning_cost}
Let $F^{init}$ be the initial forest rooted and $F'$ the forest obtained after pruning each tree using Algorithm \ref{alg:pruning}. Similarly, let $U^{init}$ be the initial set of clients that were not covered by $\mathcal P$ and $U'$ be the set of clients that remain uncovered after applying Algorithm \ref{alg:pruning} to each tree of $F^{init}$.

Let $\mathcal T_1, \mathcal T_2, \ldots, \mathcal T_m$ be the tours that were found over all applications of Algorithm \ref{alg:pruning} to each tree in $F^{init}$. Then $\sum_{i=1}^m c(\mathcal T_i) \leq 2.5 \cdot (c(F^{init})-c(F')) + \frac{1}{\beta} \cdot (\ell(U^{init})-\ell(U'))$.
\end{corollary}
\begin{proof}
Each time a tour $\mathcal T_i$ is included in the solution, its cost is bounded by $2.5$ times the cost of the subtrees that are removed plus $1/\beta$ times the radial lower bound of the uncovered clients that were just now covered by $\mathcal T_i$.
\end{proof}

We also show that the remnants of trees of the forest that remain after this step are ``small'' and satisfy an additional property relating their cost to their uncovered clients' radial lower bound contribution if they are not too small.
These properties will be helpful in the final stage of the algorithm.

\begin{lemma}\label{lem:pruning_analysis}
Let $T$ be any tree in the forest $F$ and $T'$ be the result of applying Algorithm \ref{alg:pruning} to $T$. Then $|U(T')| \leq \beta \cdot k$ and, further, if $|U(T')| > k/2$ then $\ell(U(T')) \geq \Delta \cdot c(T')$.
\end{lemma}
The proof involves a number of cases to check and is deferred to Section \ref{sec:pruning_analyze}. 


\subsection{Step: Forming the Final Tours}\label{sec:final_part}

Let $F'$ be the portion of the initial forest $F$ that remains after applying Algorithm \ref{alg:pruning} to each tree $T$ of $F$. Recall that $U$ denotes the clients that have not yet been covered by the algorithm, so at this point $U$ is the set of all clients not lying on a path in $\mathcal P$ and not covered by a tour that was produced the various calls to Algorithm \ref{alg:pruning}. Each tree $T'$ in $F'$ remains rooted in a path of $\mathcal P$ and all clients currently in $U$ lie on one of these trees. These trees $T'$ satisfy the properties in Lemma \ref{lem:pruning_analysis}.

For each path $P \in \mathcal P$, let $C_P$ denote the clients on $P$. If $U = \emptyset$ at this point, we would conclude the algorithm by taking each path in $\mathcal P$ and applying the tour splitting procedure from \cite{li1990worst}. This charges each edge on each path at most twice and plus an additional cost of $\ell(C_P)$. But there may be ``small'' trees left in the forest $F'$. One could try to graft them into the path $P$ they are rooted in naively by doubling their edges and shortcutting the Eulerian tour to get a path that can be grafted into $P$. But then when applying tour splitting to the path, we might charge an edge of a tree $T'$ up to four times. In particular, if the tour is split at a node $v$ in a tree $T'$ then the root-to-$v$ path in $T'$ would be doubled twice in the final analysis.

The main idea in this section is to handle this problem with care by selecting a careful depth-first search (DFS) ordering of each tree $T'$ and being careful to only charge a cheap portion of each tree more than twice. The key concept here is to look at how deep the tree is. For a tree $T'$ in what is left of the forest, let $v_{T'}$ be the root of $T'$ and define \texttt{trunk}($T'$) to be the highest-cost path from $v_{T'}$ to a leaf node in $T'$.

Consider a DFS of $T'$ starting at $v_{T'}$ which first recurses along the path to the leaf node defining \texttt{trunk}($T$) but all other decisions in the DFS can be made arbitrarily. Output the uncovered nodes of $T'$ in a post-order traversal in this DFS. In particular, the DFS will not output any node until after it has traveled down along the entire trunk and it will not output the root $v_{T'}$ itself.

To create the final tours we introduce one final concept. Namely, we classify a tree $T'$ in one of two ways.
\begin{itemize}
\item {\bf Short} if $c(\texttt{trunk}(T')) \leq \frac{c(T)}{2} + \frac{\ell(U(T'))}{4} \cdot \frac{|U(T')|}{k}$.
\item {\bf Tall} otherwise.
\end{itemize}

For each path $P \in \mathcal P$ with, say, $P = v_1, v_2, \ldots, v_a$ (here $v_1 \in R$) construct a final ordering of $P$ and all trees of $F'$ rooted in $P$ as follows: after $v_i$ but before $v_{i+1}$, output the nodes of the tree $T'$ of $F'$ rooted at $v_i$ in the post-order DFS traversal mentioned above. See the top part of Figure \ref{fig:pathsplit} to see an example of this ordering.
Let $P'$ be the resulting sequence $r = u_1, u_2, u_3, u_3, \ldots, u_{a'}$. Pick an offset $\tau \in \{1, 2, \ldots, a\}$ and let $S_{P'} = \{u_\tau, u_{\tau+k}, u_{\tau+2k}, \ldots \}$ be nodes about which we will {\em split} the sequence $P'$. For brevity $r_i$ be the depot nearest to $u_i$ for each $i$.

We create tours as follows. First, partition $P'$ into subsequences using this random offset, i.e., for each integer $b \geq 0$ consider the sequence $P'_b := u_{b \cdot k + \tau}, u_{b \cdot k + \tau+1}, \ldots, u_{(b+1) \cdot k + \tau - 1}$ (omitting indices outside the range $[1,a']$). So all but, perhaps, the first
and last subsequence have length exactly $k$.

For each subsequence $P'_b$, if it is not the last subsequence and if the next client $u_{(b+1) \cdot k + \tau}$ (which is in $S_{P'}$) lies in $U(T')$ for a tall tree $T'$, we further split $P'_b$ as follows. Say $u_j$ is the base of $T'$, notice $u_j$ lies in $P'_b$ because $|U(T')| \leq \beta \cdot k < k$. Split $P'_b$ into two subsequences: the prefix up to and including $u_j$ and the suffix after $u_j$. Finally, append $u_{(b+1) \cdot k + \tau - 1}$ to the second subsequence. After doing this for all $P'_b$, note the following properties: (a) each subsequence contains exactly one node in $S_{P'}$ or the depot $u_1$ at the start of $P$, (b) each node $S_{P'}$ lies in two subsequences only if it is in $U(T')$ for a tall tree $T'$, otherwise it appears on exactly one subsequence.

Finally, for each of the resulting subsequences we obtain a tour spanning by doubling all edges of $P \cup F'$ lying between clients of the subsequence and adding a doubled edge connecting the unique node of $\{u_1\} \cup S_{P'}$ in the subsequence to its nearest depot.
Shortcut the resulting tour past all nodes except the depot and those in the subsequence. If the subsequence was the suffix of a split subsequence of the form $P'_b$, then also shortcut the tour past $u_{(b+1) \cdot k + \tau}$ as it will be covered in the next subsequence. This way, every tour spans at most $k$ clients and all uncovered clients grafted onto $P$ are now covered.
This procedure is depicted in Figure \ref{fig:pathsplit}.

\begin{figure}[ht]
\begin{center}
\includegraphics[width=13cm]{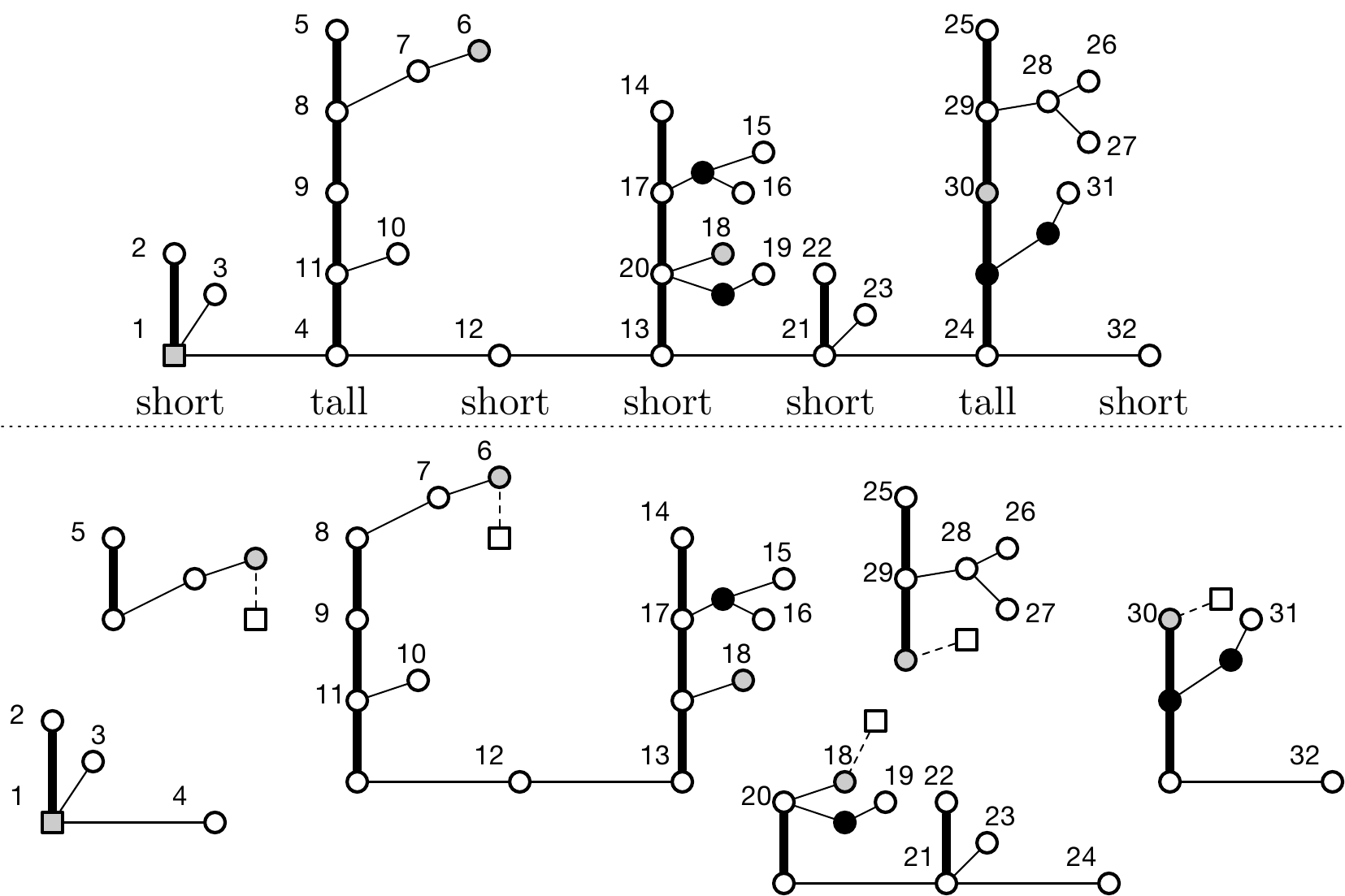}
\end{center}
\caption{{\bf Top}: A path $P \in \mathcal P$, depicted horizontally, and the trees $T'$ of $F'$ with a root in $P$. The thick vertical edges are the trunks of the trees. The numbers indicate the final ordering of the nodes. The black nodes were covered earlier by tours in the tree pruning stage.
With $k = 12$ and offset $\tau = 6$, the set consisting of the depot and $S_{P'}$, depicted with grey nodes, would be nodes $1, 6, 18$ and $30$ in this ordering.
{\bf Bottom}: The edges that were doubled to form the various tours. The white square nodes are depots nearest each node of $S_{P'}$. Notice each edge lies on at most one subtree with the following
exceptions: for a large tree $T'$ some edges from a single path between $S_{P'}$ and the trunk appear on two subtrees and for a small tree $T'$ some edges on a single path from $S_{P'}$ to the root of the tree appear on two subtrees. A client that is not numbered in means we do not consider it as being covered by the corresponding tour.}
\label{fig:pathsplit}
\end{figure}

To analyze the cost of this tour, for each subsequence we double all edges of $P$ and the various trees $T'$ rooted in $P$ that lie between two nodes of the subsequence. Intuitively, every edge of $P$ is used by at most one subsequence and every edge of each tree $T'$ is as well with two exceptions:
\begin{itemize}
    \item In short trees $T'$, if some node of $U(T')$ lies in $S_{P'}$ then edges lying on some path to the root $v_{T'}$ will be used by two subsequences. This happens only with probability $|U(T')|/k$ and if it does, the cost of the edges that are used by two subsequences is at most $c(\texttt{trunk}(T'))$.
    \item In tall trees $T'$, if some node of $U(T')$ lies in $S_{P'}$ then edges lying on some path from $S_{P'}$ to a node on $\texttt{trunk}(T')$ will be used by two subsequences. This happens only with probability $|U(T')|/k$ and if it does, the cost of the edges that are used by two subsequences is at most $c(T') - c(\texttt{trunk}(T'))$. But in this case we also have to remember that the node $v \in U(T') \cap S_{P'}$ will be used to connect two subsequences to its nearest depot (i.e. we will charge $\ell_v$ twice in this case).
\end{itemize}
The precise analysis is summarized below, the proof appears in Section \ref{sec:splitting_analyze}.
\begin{lemma}\label{lem:splitting}
Let $c_{F'}(P)$ be the total cost of edges of $F'$ in trees $T'$ with a root in $P'$ and let $U_{F'}(P)$ be the uncovered clients lying in these trees.
For some choice of offset $\tau$, the total cost of these tours is at most $2 \cdot c(P) +2.5 \cdot c_{F'}(P) +  \ell(C_P) + \frac{1}{\beta} \cdot \ell(U_{F'}(P))$.
\end{lemma}
So the algorithm for this section is to try all offsets $\tau$ and keep the tours that arise from the choice that leads to the cheapest solution.
Summing this bound over all $P \in \mathcal P$ yields the following.
\begin{corollary}\label{cor:splitting_cost}
Let $C'$ denote all clients lying on some path in $\mathcal P$, $F'$ and $U'$ denote the resulting forest and, respectively, the set of clients not covered after applying Algorithm \ref{alg:pruning} to each tree in the initial forest $F^{init}$.
Applying this tour-splitting procedure to all paths $P \in \mathcal P$ produces tours covering all clients in $C' \cup U'$ using tours with total cost at most $2 \cdot c(\mathcal P) + 2.5 \cdot c(F') + \ell(C') + \frac{1}{\beta} \cdot \ell(U')$.
\end{corollary}


\subsection{Approximation Guarantee Analysis}
For this section, we revert to letting $U$ be the set of clients initially not covered by the paths $\mathcal P$.
Considering both the tree pruning and final tours obtained by splitting paths in $\mathcal P$, Corollaries \ref{cor:pruning_cost} and \ref{cor:splitting_cost} show the total cost of all tours is at most
\[ 2 \cdot c(\mathcal P) + 2.5 \cdot c(F) + \ell(C-U) + \frac{1}{\beta} \cdot \ell(U) = 2 \cdot c(\mathcal P) + 2.5 \cdot c(F) + \LB + \left(\frac{1}{\beta} - 1\right) \cdot \ell(U). \]

Taking this in expectation over the random sampling of branchings and using Lemmas \ref{lem:path}, \ref{lem:penalty} and \ref{lem:bridge} the expected cost of the final solution is at most
\[ 2 \gamma \cdot (1 + \delta) \cdot \opt_{LP} + 2.5 \cdot e^{-\gamma} \cdot \opt_{LP} + \LB + \left(\frac{1}{\beta} - 1\right)\cdot e^{-\gamma} \cdot  \LB.\]
Using $\LB \leq (1-\delta) \cdot \opt_{LP}$ (Lemma \ref{lem:lb}), the approximation guarantee (relative to $\opt_{LP}$) is at most
\[ 2 \gamma \cdot (1+\delta) + \left(2.5 + (1-\delta) \cdot (1/\beta - 1))\right) \cdot e^{-\gamma} + (1-\delta). \]

To help us choose the best constants, we considered when $\delta = 0$ and optimize for this case. The guarantee is of the form $A \gamma + B e^{-\gamma} + C$ which is minimized by selecting $\gamma = \ln \frac{B}{A}$. In the case $\gamma = 0$, we have $A = 2, B = 1.5 + 1/\beta$ and $C = 1$. Using $\beta = 0.59024$ has us set $\gamma := 0.46821$. We remark that not every choice of $\beta$ can be considered as the proofs from previous sections only work if $\beta$ can be chosen such that there is some $\Delta$ can be satisfying the bounds stated in Section \ref{sec:pruning}. Plugging this back in to the approximation guarantee for arbitrary $\delta$, we see it is, in expectation, at most $3.9365 - 0.49826 \cdot \delta$.

\subsection{Compatibility with \cite{blauth2021improving} and Further Improvements}
Suppose $\delta'$ is defined in so that $\LB = (1-\delta') \cdot \opt$.
If we add the constraint $\sum_{r,v} 2c(v,R) \leq (1-\delta') \cdot \opt$ in our LP relaxation\footnote{By a standard scaling argument, we can assume edge costs are polynomially-bounded integers which would then allow us to guess the value of $\delta'$ exactly while only losing an arbitrarily-small constant $\epsilon$ in the guarantee.}, it would still be valid to say our algorithm finds a solution whose cost is at most $(3.9365 - 0.49826 \cdot \delta') \cdot \opt$.

If $\delta'$ can be bounded from below by a small constant, our approximation guarantee is slightly better than $3.9365$. On the other hand, for the function $f(\delta)$ in \cite{blauth2021improving}, the tour splitting algorithm for \cvrpmd can be modified to find a solution with an approximation guarantee of $2 \cdot (1+f(\delta)) + 1$. The function $f(\delta)$ approaches $0$ as $\delta \rightarrow 0$, but it does so fairly slowly. Still, it eventually becomes small enough when $\delta$ is bounded above by a small constant so that taking the better of our algorithm and the tour splitting modification from \cite{li1990worst} will improve the approximation guarantee a bit further, though the improvement will only be in the order of $10^{-3}$ as it was in \cite{blauth2021improving} and \cite{ZX23}. We have not computed this marginally-improved constant.

Finally, one could get a bit more out of this idea and choose $\beta$ and $\Delta$ optimally in terms of $\delta$ as well but this idea seems to only get very slight improvements (certainly smaller than $10^{-3}$). 



\section{Completing the Analysis}\label{sec:analysis}

\subsection{Proof of Lemma \ref{lem:pruning_analysis}}\label{sec:pruning_analyze}
We prove this by contradiction, namely we show if $T'$ does not satisfy these bounds then the algorithm should have found another tour rather than terminating with $T'$.

Let $u$ be the deepest node with $|U(T'_v)| > k$, letting $u = v_{T'}$ (the root of $T'$) if there is no such node, i.e. the current tree $T'$ has at most $k$ uncovered nodes. Let $G_1, G_2, \ldots$ be the groups formed form the subtrees rooted at children of $u$. So $|U(G_i)| \leq k$ for each group $G_i$ yet $|U(G_i)| + |U(G_j)| > k$ for any two distinct groups $G_i$ and $G_j$.

We first establish that an individual group would satisfy the bounds from the lemma statement.
\begin{claim}\label{claim:groupbound}
For any group $G_i$, we have $|U(G_i)| \leq \frac{1}{\beta} \cdot k$ and if $|U(G_i)| > k/2$ then $\ell(U(G_i)) \geq \Delta \cdot c(G_i)$.
\end{claim}
\begin{proof}
Note for any group $G_i$ that $\mathcal T(G_i)$ was first formed by doubling all edges in subtrees in $G_i$ plus their parent edges to $u$ and then adding a cheapest doubled edge connecting some depot in $R$ to some uncovered client in $U(G_i)$. So
\[ c(\mathcal T(G_i)) \leq 2 \cdot c(G_i) + \frac{k}{|U(G_i)|} \cdot \ell(U(G_i)) \]
because $\min_{w \in U(G_i)} 2 \cdot c(w,R) \leq \frac{1}{|U(G_i)|} \sum_{w \in U(G_i)} 2 \cdot c(w,R) = \frac{k}{|U(G_i)|} \cdot \ell(U(G_i))$.
So if $|U(G_i)| \geq \beta \cdot k$ then $\mathcal T(G_i) \leq 2 \cdot c(G_i) + \frac{k}{|U(G_i)|} \ell(U(G_i)) \leq 2 \cdot c(G_i) + \frac{1}{\beta} \cdot \ell(U(G_i))$
meaning the algorithm would have found another tour, a contradiction.

Similarly, if $|U(G_i)| > k/2$ yet $\ell(U(G_i)) \leq \Delta \cdot c(T')$ then
\[ c(\mathcal T(G_i)) \leq 2 \cdot c(G_i) + 2 \cdot \ell(U(G_i)) \leq 2.5 \cdot c(G_i) + (2 - 0.5/\Delta) \cdot \ell(U(G_i)) \leq 2.5 \cdot c(G_i) + \frac{1}{\beta} \cdot \ell(U(G_i)) \]
where we recall $2-0.5/\Delta \leq 1/\beta$. So again the algorithm would have found another tour. A contradiction.
\end{proof}

\begin{claim}
There are at least two groups.
\end{claim}
\begin{proof}
Suppose otherwise, i.e., that there is a single group $G_1$. By the previous claim, $|U(G_1)| \leq \frac{1}{\beta} \cdot k$. But then $|U(T_u)| \leq 1 + \frac{1}{\beta} \cdot k \leq k$ (because $1/(1-\beta) < 3 \leq k$). The only way for this to hold is if $u = v_{T'}$. But then $U(T') = U(G_1)$, recalling $v_{T'}$ is covered as it lies on a path in $\mathcal P$. So $c(T') = c(G_1)$ since the cost of a group includes the parent edges from its subtrees to $u$. Since $G_1$ satisfies the bounds from Claim \ref{claim:groupbound}, then $T'$ satisfies the bounds in the lemma statement which is a contradiction. So there are at least two groups.
\end{proof}

Let $G, G'$ be any two groups ordered so that $c(G) \geq c(G')$. Let $A$ be the set of $k - |U(G)|$ clients in $U(G')$ with largest $\ell_v$ values. Notice
\begin{equation}\label{eqn:twogroup}
c(\mathcal T(G,G')) \leq 2 \cdot c(G) + 2 \cdot c(G') + \min_{v \in U(G) \cup A} c(v,R).
\end{equation}
We know $|U(G)| + |U(G')| > k$ since we cannot merge the groups, so at least of $|U(G)|$ or $|U(G')|$ is greater than $k/2$. The analysis continues with these two cases.

~

\noindent
{\bf Case 1}: $|U(G)| > k/2$\\
If so, we can further bound \eqref{eqn:twogroup} using the following approach. First, the minimum of $2 \cdot c(v,R)$ can be bounded by the expected value of $2 \cdot c(u,R)$ if we sample $u$ as follows: with some probability $p$ we sample $v \sim U(G)$ uniformly at random and with probability $1-p$ we sample $v \sim A$ uniformly at random. Since the minimum is at most the expected value,
\begin{eqnarray*}
\min_{v \in U(G) \cup A} 2 \cdot c(v,R)
& \leq & p \cdot \frac{1}{|U(G)|} \sum_{v \in U(G)} 2 \cdot c(v,R) + (1-p) \cdot \sum_{v \in A} 2 \cdot c(v,R) \\
& = & p \cdot \frac{k}{|U(G)|} \cdot \ell(U(G)) + (1-p) \cdot \frac{k}{|A|} \cdot \ell(A) \\
& \leq & 2p\cdot \ell(U(G)) +  \frac{1-p}{1-\beta} \cdot \ell(A).
\end{eqnarray*}
The last bound is because $k/2 \leq |U(G)| \leq \beta \cdot k$ so $|A| = k-|U(G)| \geq (1-\beta)\cdot k$.
Using this, we can continue the bound from \eqref{eqn:twogroup} as follows:
\begin{eqnarray*}
c(\mathcal T(G,G')) & \leq & 2 \cdot c(G) + 2 \cdot c(G') + 2p \cdot \ell(U(G)) + \frac{1-p}{1-\beta} \cdot \ell(A) \\
& \leq & 4 \cdot c(G) + 2p \cdot \ell(U(G)) + \frac{1-p}{1-\beta} \cdot \ell(A) \\
& \leq & 2.5 \cdot c(G) + \left(\frac{1.5}{\Delta} + 2p\right) \cdot \ell(U(G)) + \frac{1-p}{1-\beta} \cdot \ell(A)
\end{eqnarray*}
The last bound is because $c(G) \leq \ell(U(G)) / \Delta$ holds by Claim \ref{claim:groupbound}.

Setting $p = \frac{\Delta - (1-\beta) \cdot 1.5}{\Delta \cdot (3-2\beta)} \approx 0.34302$ balances the two parameters and has the coefficients of $\ell(U(G))$ and $\ell(A)$ both be smaller than $\frac{1}{\beta}$. So $c(\mathcal T(G, G')) \leq 2.5 \cdot c(G) + \frac{1}{\beta} \cdot \ell(U(G) \cup A)$ which contradicts the fact that no tour would have been formed from $T'$.

~

\noindent
{\bf Case 2}: $|U(G)| \leq k/2$\\
We note that both $G$ and $G'$ satisfy the properties from Claim \ref{claim:groupbound} and that $|U(G')| > k/2$ since $|U(G)| \leq k$. Also, since $|U(G')| \leq \beta \cdot k$ then we have $|U(G)| \geq (1-\beta) \cdot k$.

~

\begin{claim}\label{claim:shift}
$\ell(U(G')) \leq 2\beta \cdot \ell(A)$.
\end{claim}
\begin{proof}
Since $A$ are the clients in $U(G')$ with largest $\ell_v$-values, then $\ell(A) \geq \frac{|A|}{|U(G')|} \cdot \ell(U(G')) = \frac{k-|U(G)|}{|U(G')|} \cdot \ell(U(G'))$. We know $|U(G)| \leq k/2$ and $|U(G')| \leq $, so $\frac{k-|U(G)|}{|U(G')|} \geq 1/(2\beta)$ as required.
\end{proof}

We continue bounding \eqref{eqn:twogroup} as in the previous case except we use a different probability $p$.

\begin{eqnarray*}
\min_{v \in U(G) \cup A} 2 \cdot c(v,R)
& \leq & p \cdot \frac{1}{|U(G)|} \sum_{v \in U(G)} 2 \cdot c(v,R) + (1-p) \cdot \sum_{v \in A} 2 \cdot c(v,R) \\
& = & p \cdot \frac{k}{|U(G)|} \cdot \ell(U(G)) + (1-p) \cdot \frac{k}{|A|} \cdot \ell(A) \\
& \leq & \frac{p}{1-\beta} \cdot \ell(U(G)) + 2 \cdot (1-p) \cdot \ell(A)
\end{eqnarray*}

From this, we have
\begin{eqnarray*}
c(\mathcal T(G,G')) & \leq & 2 \cdot c(G) + 2 \cdot c(G') + \frac{p}{1-\beta} \cdot \ell(U(G)) + 2 \cdot (1-p) \cdot \ell(A) \\
& \leq & 2.5 \cdot c(G) + 1.5 \cdot c(G') + \frac{p}{1-\beta} \cdot \ell(U(G)) + 2 \cdot (1-p) \cdot \ell(A) \\
& \leq & 2.5 \cdot c(G) + \frac{1.5}{\Delta} \cdot \ell(U(G')) + \frac{p}{1-\beta} \cdot \ell(U(G)) + 2 \cdot (1-p) \cdot \ell(A) \\
& \leq & 2.5 \cdot c(G) + 2\beta \cdot \frac{1.5}{\Delta} \cdot \ell(A) + \frac{p}{1-\beta} \cdot \ell(U(G)) + 2 \cdot (1-p) \cdot \ell(A) \\
& = & 2.5 \cdot c(G) +  \frac{p}{1-\beta} \cdot \ell(U(G)) + \left(2 \cdot (1-p) +\frac{3\beta}{\Delta}\right) \cdot \ell(A) 
\end{eqnarray*}
Setting $p = \frac{2 + 3\beta/\Delta}{2 + 1/(1-\beta)} \approx 0.69425$ balances these coefficients and has them both be at most $1/\beta$. Again this means $c(\mathcal T(G, G')) \leq 2.5 \cdot c(G) + \frac{1}{\beta} \cdot \ell(U(G) \cup A)$ which contradicts the fact that no tour would have been formed from $T'$.

To conclude, we have shown that if the algorithm terminated with $T'$ not satisfying the properties in the statement of Lemma \ref{lem:pruning_analysis} then the groups considered in the final iteration that created no tours would satisfy the bounds in Claim \ref{claim:groupbound} and there were at least two groups. But the previous two cases show that a tour would indeed have been created in the second half of the body of the while loop in Algorithm \ref{alg:pruning}, a contradiction.
This completes the proof of Lemma \ref{lem:pruning_analysis}

\subsection{Proof of Lemma \ref{lem:splitting}}\label{sec:splitting_analyze}

Recall $C_P$ denotes the clients on $P$ and $P'$ was an ordering of nodes of $P$ and its corresponding trees $T'$ in $F'$. For each subsequence $\sigma$ that $P'$ was partitioned into, let $E_\sigma$ denote all edges of $P$ and $F'$ that lie between at least one pair of nodes in the subsequence. The following can be carefully verified. Here, every tree $T'$ mentioned is a tree in $F'$ whose root lies on $P$.
\begin{itemize}
\item Every edge of $P$ lies in at most one set of the form $E_\sigma$.
\item For a tall tree $T'$, every edge of $T'$ lies in at most one set of the form $E_\sigma$ with one exception: edges on a path from $v$ to $\texttt{trunk}(T')$ for some $v \in U(T') \cap S_{P'}$ lie in at most two sets of the form $E_{\sigma}$.
\item For a short tree $T'$, every edge of $T'$ lies in at most one set of the form $E_\sigma$ with one exception: edges on a path from $v$ to the root of $T'$ for some $v \in U(T') \cap S_{P'}$ lie in at most two sets of the form $E_{\sigma}$.
\item No tree $T'$ contains more than one node in $S_{P'}$.
\end{itemize}

The cost of the tours produced by this procedure is at most $\sum_\sigma 2 \cdot c(E_\sigma) + \sum_{v \in S_{P'}} 2 \cdot c(v,R) + \sum_{T' \texttt{ large}} \sum_{v \in S_{P'} \cap U(T')} 2 \cdot c(v,R)$.
The last term is because a vertex in $U(T')$ for a large tree will connect to its nearest depot twice if it is in $S_{P'}$.

For a short tree $T'$, if some $v \in U(T')$ lies in $S_{P}$ then the only edges lying in two different $E_{\sigma}$ lie on a path of cost at most $c(\texttt{trunk}(T'))$ since every path in $T'$ is at most the cost of the trunk.
For a tall tree $T'$, if some $v \in U(T')$ lies in $S_{P}$ then the only edges lying in two different $E_{\sigma}$ lie on a path of cost at most $c(T') - c(\texttt{trunk}(T'))$ since the edges of this path are disjoint from the trunk.

If we pick $\tau$ randomly, each $v \in C_P \cup \sum_{T'} U(T')$ lies in $S_{P'}$ with probability exactly $1/k$. So for a short tree, the probability we charge an extra set of edges of cost at most $c(\texttt{trunk}(T')$ is $\frac{|U(T')|}{k}$
and for a tall tree the probability we charge an extra set of edges of cost at most $c(T') - c(\texttt{trunk}(T'))$ is also $\frac{|U(T')|}{k}$. So the cost of the tours for the best $\tau$ is at most the expected cost of the tours over the random choice of $\tau$, which is at most
\[ 2 \cdot c(P) + \ell(C_P) + \sum_{T' : \text{short}} \left[2 c(T') + 2 \cdot \frac{|U(T')|}{k} \cdot c(\texttt{trunk}(T')) + \frac{|U(T')|}{k} \cdot \ell(U(T'))\right] \]
\[+ \sum_{T : \text{tall}} \left[ 2c(T') + 2\cdot \frac{|U(T')|}{k} \cdot (c(T') - c(\texttt{trunk}(T'))) + 2\cdot \frac{|U(T')|}{k} \ell(U(T')) \right]. \]

\begin{lemma}
For each tree $T'$ rooted in $P$, its corresponding term in the previous sum is at most $(2 + |U(T')|/k) \cdot c(T) + 1.5 \cdot \ell(U(T'))$.
\end{lemma}
\begin{proof}
For brevity in this proof, let $t = c(\texttt{trunk}(T')), c = c(T'), \ell = \ell(U(T'))$ and $p = |U(T')|/k$. For small $T$, this means $t \leq c/2 + \ell/(4p)$ and the term in the sum for $T$ is then bounded by
\[ 2c + 2pt + \ell \leq 2c + 2p(c/2 + \ell/(4p)) + p\ell = (2+p)c + 1.5 \ell. \]
For tall trees we have $t \geq c/2 + p\ell/4$ and the term in the sum for $T$ is then bounded by
\[ 2c + 2p(c-t) + 2p\ell \leq 2c + 2p(c-c/2-\ell/(4p)) + 2\ell = (2+p)c + 1.5\ell. \]
\end{proof}

Summarizing, the cost of tours constructed in this step for a path $P$ and its associated trees $T$ is at most
\[ 2 c(P) + \ell(C_P) + \sum_T (2+|U(T')|/k) c(T') + 1.5 \ell(U(T')) \]
We still need to improve this a bit since the coefficient in front of $c(T')$ can be as much as $2+\beta$ whereas we need it to be a most $2.5$ in our final analysis.

For a tree $T'$ with $|U(T')|/k \leq 1/2$ the coefficient of $c(T')$ is already at most $2.5$. For a tree with $|U(T')|/k > 1/2$, we do still have $|U(T')|/k \leq \beta$ by Lemma \ref{lem:pruning_analysis} and, further, $c(T') \leq \ell(U(T')) / \Delta$.
So in this case we can bound the contribution of $T$ to the cost of the tours by
\[ (2 + \beta) c(T') + 1.5 \ell(U(T')) \leq (2 + 0.5) c(T') + (1.5 + (\beta-0.5)/\Delta) \ell(U(T')) \leq 2.5 \cdot c(T') + \frac{1}{\beta} \cdot \ell(U(T')). \]
Summing over all $T'$ rooted in $P$, the cost of all tours produced by splitting $P$ and its trees in this way is at most
\[ 2 c(P) + \ell(C_P) + \sum_T 2.5 c(T) + \frac{1}{\beta} \ell(U_T). \]
This completes the proof of Lemma \ref{lem:splitting}.


\bibliographystyle{alpha}
\bibliography{refs}

\end{document}